\documentclass[fp]{jpsj3x}
\usepackage{txfonts}
\usepackage{graphics}
\usepackage[dvipdfmx]{graphicx}
\usepackage{amsmath}
\usepackage{amssymb}
\usepackage{amscd}
\usepackage[mathscr]{eucal}
\usepackage{bm}
\usepackage{simplewick}
\usepackage{color}
\newcommand{\red}[1]{\textcolor{black}{#1}}
\newcommand{\N}{{\mathbb N}}
\newcommand{\Z}{{\mathbb Z}}

\newcommand{\C}{{\mathbb C}}
\def\normOrd#1{\mathop{:}\nolimits\!#1\!\mathop{:}\nolimits}
\def\normOrdx#1{\mathop{:}\nolimits\!#1\!\mathop{:}\nolimits}
\newtheorem{theorem}{Theorem}
\newtheorem{lemma}{Lemma}

\newtheorem{proposition}{Proposition}

\newtheorem{definition}{Definition}

\newtheorem{remark}{Remark}

\newcommand{\proof}{\noindent \textit{Proof. }}
\newcommand{\qed}{\hfill $\Box$}
\title{Generalized Wick Theorems
in Conformal Field Theory and the Borcherds Identity}

\author{Taichiro Takagi, Takuma Yoshikawa}
\inst{Department of Applied Physics, National Defense Academy,
Kanagawa 239-8686, Japan} 

\abst{As \red{a} counterpart of the well-known generalized Wick theorem \red{by Bais et.~al.~in 1988} for interacting fields in 
two dimensional conformal field theory, we present
a new \red{contour integral} formula for the operator product expansion
of a normally ordered operator and  a single operator
{\em on its right hand}.
Quite similar to the original Wick theorem
for the opposite order operator product,
it 
expresses the contraction i.~e.~the singular part of the operator product expansion as a contour integral of only two terms,
each of which is a product of a
contraction and a single operator.
We discuss the relation between these formulas and the Borcherds identity
satisfied by the quantum fields associated with the 
theory of vertex algebras. }

\begin{document}
\maketitle

\section{Introduction}\label{sec:1}
More than thirty years two dimensional (2D) conformal field theory (CFT) has been widely
accepted as a distinguished theory for giving mathematical foundations of
string theory and for providing a language to neatly describe 2D critical phenomena in
statistical physics.
As one of many simple and beautiful formulas in CFT
the following one is known as the generalized Wick theorem 
for interacting fields \cite{BBSS88,STP88,DMS97},
and has been used extensively in many literatures:
\begin{equation}
\contraction{}{A}{(z)(}{BC}
A(z)(BC)(w) 
= \frac{1}{2 \pi \sqrt{-1}}
\oint_{C_w} \frac{d x}{x-w} \left\{\contraction{}{A}{(z)}{B}
A(z)B(x) C(w)+ B(x) \contraction{}{A}{(z)}{C}A(z) C(w) \right\}. \label{eq:july3_1}
\end{equation}
Here, $A(z)$ etc.~are operators for chiral conformal fields,
$(BC)(w) = \normOrd{B(w)C(w)}$ denotes the normally ordered product,
{$\contraction{}{A}{(z)}{B}A(z)B(x)$} denotes the contraction, i.~e.~the singular part of the operator product expansion (OPE), and
$C_w$ is a contour encircling the point $w$ with an infinitesimal radius.
Also, the radial orderings of the operators are implicitly assumed.

It is a natural question to ask whether an {\em analogous} expression exists if we contract with the normally ordered operator on the left,
\red{where distinction of the left and the right is significant because we are dealing with interacting fields.}
In order to make the meaning of the question clearer, 
we want to specify the meaning of the word {\em analogous} here.
The formula \eqref{eq:july3_1} is describing
a relation among three contractions, and a contraction is a sort of
generating functions of the coefficient operators arisen in it.
The use of contractions has brought \eqref{eq:july3_1} to being an expression that most physicists are more familiar with than the other expressions
\cite{Th95, Kac96, MN97}. 
In this sense, as far as the authors know no such integral expression 
analogous to \eqref{eq:july3_1} \red{for the opposite order OPE}
has been written in
any literature, admitting that there have been equivalent formulas by a series of relations among the coefficient operators themselves \cite{Th95, MN97},
or by using another type of generating functions
that are not the contractions \cite{BK03, KRW03}.
The purpose of this paper is to answer this question
by showing that there is certainly such an integral formula analogous to \eqref{eq:july3_1} which is given by
\begin{equation}
\contraction{}{(AB)}{(z)(}{C}
(AB)(z)C(w) 
= \frac{1}{2 \pi \sqrt{-1}}
\oint_{C_w} \frac{d x}{z-x} \left\{
\normOrdx{A(x) \contraction{}{B}{(x)}{C}B(x)C(w)} + B(x) \contraction{}{A}{(x)}{C}A(x) C(w) \right\}. \label{eq:main}
\end{equation}
Here we note that the expansion coefficients of
$\contraction{}{B}{(x)}{C}B(x)C(w)$ and $\contraction{}{A}{(x)}{C}A(x) C(w)$
are fields at the point $w$,
hence it makes sense to consider operator products of a field at the point $x$ such as 
$A(x)$ or $B(x)$ and them.
The normal ordering $\normOrdx{ \quad }$ for the first term in the integrand
is to subtract the singular terms from such 
OPEs.
We will discuss relations of these formulas with 
specializations of the Borcherds identity
that is an axiom for a vertex algebra \cite{B86}.

The main results of this paper are Theorems \ref{th:2} 
and \ref{th:4}.
While an algebraic method is used for the former, we use a rigorous 
analytic method developed in Ref.~\citen{MN97} for the latter.
Besides it enables us to
derive the Wick theorems in a heuristic way,
we have another reason for dealing with the analytic method. 
That is, in a sense one had been forced to use the other types of
expressions \cite{MN97, BK03, KRW03} for describing those theorems 
rigorously.
Probably,
this is because of the difficulty of giving a mathematically rigorous 
definition of the integrals of operators themselves
with respect to their arguments viewed as integration variables.
In this paper we try to describe
the generalized Wick theorems 
not only 
in a way that respects its original contour integral
expression as in Ref.~\citen{BBSS88},
but also in a way that respects mathematical rigor by means of 
the above mentioned analytic method.

This paper is organized as follows:
In Sect.~\ref{sec:2} we review an algebraic formulation of 2D chiral quantum fields based on Matsuo and Nagatomo \cite{MN97}.
The purpose of this section is to introduce several notions for
giving a mathematically rigorous formulation of
the generalized Wick theorems \eqref{eq:july3_1} and \eqref{eq:main}.
%
In Sect.~\ref{sec:3} we show that
the generalized Wick theorems are equivalent to
special cases of the Borcherds identity.
So far, the contour integrals of the operator valued functions
appearing in \eqref{eq:july3_1} and \eqref{eq:main} are interpreted as
formal symbols for algebraic manipulations, rather than as
rigorous analytic calculations.
In contrast, in Sect.~\ref{sec:4} we consider matrix elements of the 
product of the operators and interpret these integrals as
analytic calculations.
%
%
%
We give several discussions in Sect.~\ref{sec:6}.
In Appendix~\ref{app:A} we briefly 
discuss the relation between our formulas \eqref{eq:july3_1}, \eqref{eq:main} and 
the formulas expressed by
another type of generating functions
in Refs. \citen{BK03, KRW03}.
In Appendix~\ref{sec:5} the usefulness of our new formula is illustrated by a few examples

\section{An Algebraic Formulation of Two-dimensional Chiral Quantum Fields
and Their Operator Product Expansions}\label{sec:2}
\subsection{Formal series, fields and their residue products}\label{sec:2_1}
In order to express the statement of the Wick theorems in a mathematically rigorous manner, 
we want to present an algebraic formulation of 2D chiral quantum fields.
Though similar formulations are provided in many literatures \cite{FBZ04,FLM88, Kac96,Kac15,MS08}, we adopt the one
by Matsuo and Nagatomo \cite{MN97} which we shall briefly review in this section.

First we introduce the notions of a field and a normally ordered product.
Let
\begin{equation}
A(z) = \sum_{n \in \mathbb{Z}} A_n z^{-n-1}\label{eq:july6_1}
\end{equation}
be a formal series in which the coefficients $A_n$ are linear transformations
on some vector space $M$.
The set of all such series is denoted by $({\rm End}\, M) [[z,z^{-1}]]$.
We call $A(z)$ a field (of one variable) if for any $b \in M$ there exists $n_0 \in \Z$ such that
$A_n b = 0$ for all $n \geq n_0$.
Also for any set of linear transformations
$\{ A_{p,q} \}_{p,q \in \mathbb{Z}}$ we consider the following formal series
\begin{equation}\label{eq:sep4_2}
A(y,z) = \sum_{p,q \in \mathbb{Z}} A_{p,q} y^{-p-1} z^{-q-1},
\end{equation}
and denote the set of all such series by $({\rm End}\, M) [[y,y^{-1},z,z^{-1}]]$.
We call $A(y,z)$ a field of two variables if for any $v \in M$
there exist $p_0, q_0 \in \mathbb{Z}$ such that if  $p \geq p_0$ or $q \geq q_0$ then $A_{p,q} v = 0$
(namely, if $p \geq p_0$ then $A_{p,q} v = 0$, and if $q \geq q_0$ then $A_{p,q} v = 0$).
If $A(y,z)$ is a field of two variables, then $A(z,z)$ is a field of one variable.
Actually we have
\begin{displaymath}
A(z,z) = \sum_{n \in \mathbb{Z}} \left( \sum_{p \in \mathbb{Z}}A_{p,n-1-p} \right)  z^{-n-1},
\end{displaymath}
that makes sense as a linear transformation on $M$.
This is because for any $v \in M$ we have $A_{p,n-1-p} v=0$ for large enough $|p|$, \red{and
for any $v \in M$ we have $\sum_{p \in \mathbb{Z}} A_{p,n-1-p} v =0$ for large enough $n$
because $p \geq p_0$ or $n-1-p \geq q_0$ is satisfied for any $p \in \mathbb{Z}$
in the subscripts of $A_{p,n-1-p}$.}

Even if both $A(y)$ and $B(z)$ are fields,
their product $A(y)B(z)$ is not necessarily a field.
However we have another type of product which turns out to be a field:
\begin{definition}
Given two series $A(y)$ and $B(z)$,
their \textbf{normally ordered product} is
\begin{equation}
\normOrd{A(y)B(z)} = A(y)_{-} B(z) + B(z) A(y)_{+},
\label{eq:july3_3}
\end{equation}
where
$A(y)_{+}=\sum_{n \geq 0} A_n y^{-n-1}, A(y)_{-}=\sum_{n < 0} A_n y^{-n-1}$.
\end{definition}
\red{It is a field of two variables if both $A(y)$ and $B(z)$ are fields,
and hence
\begin{equation*}
(AB)(z) = \normOrd{A(z)B(z)},
\end{equation*}
is a field of one variables.}
We call both $\normOrd{A(y)B(z)}$ and $\normOrd{A(z)B(z)}$ normally ordered products.

Now we introduce the notion of a residue product.
It is a generalization of the normally ordered product and
plays an essential role in this paper. 
\begin{definition}[Ref.~\citen{MN97}, Definition 1.4.1.]
Given two series $A(y)$ and $B(z)$,
their \textbf{ $\bm{m}$-th residue product} ($m \in \Z$) is
\begin{equation}
A(z)_{(m)}B(z) =
{\rm Res}_{y=0} A(y)B(z) (y-z)^m \big|_{|y|>|z|}-
{\rm Res}_{y=0} B(z)A(y) (y-z)^m \big|_{|y|<|z|},
\label{eq:oct20_1}
\end{equation}
where
\begin{eqnarray*}
(y-z)^m \big|_{|y|>|z|} &=& \sum_{i=0}^\infty (-1)^i { m \choose i } y^{m-i} z^i,\\
(y-z)^m \big|_{|y|<|z|}  &=& \sum_{i=0}^\infty (-1)^{m+i} { m \choose i } y^{i} z^{m-i}.
\end{eqnarray*}
\end{definition}
\begin{remark}
The $(y-z)^m \big|_{|y|>|z|}$ and $(y-z)^m \big|_{|y|<|z|}$ 
are formal series obtained by expanding the rational function $(y-z)^m$
into those convergent in the region $|y| > |z|$ and $|z|>|y|$ respectively. 
\end{remark}
One can prove that if both $A(z)$ and $B(z)$ are fields, then 
each term of the right hand side of \eqref{eq:oct20_1} becomes a field and hence
so is $A(z)_{(m)}B(z)$ \cite{MN97}.
In particular we have
\begin{equation*}
A(z)_{(-1)}B(z) =\normOrd{A(z)B(z)}.
\end{equation*}
In other words the $-1$th residue product is the
normally ordered product.
\subsection{Operator product expansion}\label{sec:2_2}
Now we present a mathematical formulation of the operator product
expansion of the quantum fields.
As we will see, the Wick formula \eqref{eq:july3_1} is
proved to be equivalent to a special case of the Borcherds identity without the assumption of locality (Theorem \ref{th:1}).
So, in this subsection we introduce the notion of OPE without it \cite{LZ95}.
The following argument is essentially the same as those in the proof
of Proposition 2.3 in Ref.~\citen{LZ95}, but for the
purpose of confirming our notations we present it here.

When $m \geq 0$, 
we have $(y-z)^m \big|_{|y|>|z|} =(y-z)^m \big|_{|y|<|z|} = \sum_{i=0}^m {m \choose i}
y^i (-z)^{m-i}$, hence \eqref{eq:oct20_1} is written as
\begin{eqnarray*}
A(z)_{(m)}B(z) 
&=&{\rm Res}_{y=0}  [A(y), B(z)] (y-z)^m \\
&=& [\sum_{i=0}^m {m \choose i} {\rm Res}_{y=0} A(y) y^i (-z)^{m-i}, B(z)]\\
&=& \sum_{i=0}^m {m \choose i} (-z)^{m-i}[A_i, B(z)],
\end{eqnarray*}
therefore
\begin{equation}
\frac{A(z)_{(m)}B(z)}{z^m} = \sum_{i=0}^m {m \choose i} (-1)^{m-i}\frac{[A_i, B(z)]}{z^i}.
\end{equation}
This relation is inverted as (e.g. Ref.~\citen{Stanley97}, p.~66)
\begin{equation}
\frac{[A_m,B(z)]}{z^m} = \sum_{i=0}^m {m \choose i} \frac{A(z)_{(i)} B(z)}{z^i}.
\end{equation}
By this relation we have
\begin{eqnarray*}
[ A(y)_+, B(z)] &=& \sum_{m=0}^\infty y^{-m-1}[A_m,B(z)]\\
&=& \sum_{m=0}^\infty y^{-m-1} \sum_{i=0}^m
 {m \choose i} A(z)_{(i)} B(z) z^{m-i} \\
&=& 
 \sum_{i=0}^\infty A(z)_{(i)} B(z) \sum_{m=i}^\infty  {m \choose i} y^{-m-1}z^{m-i},
\end{eqnarray*}
in which the inner summation in the last expression can be written as
\begin{equation*}
\sum_{m=i}^\infty  {m \choose i} y^{-m-1}z^{m-i} =
\partial_z^{(i)} \sum_{m=0}^\infty  y^{-m-1}z^{m} =
(y-z)^{-i-1}\big|_{|y|>|z|}.
\end{equation*}
Here, $\partial_z^{(i)} = \frac{1}{i !} \partial_z^i$.
Since the commutator $[ A(y)_+, B(z)]$ is equal to the difference of $A(y)B(z)$
and $\normOrd{A(y)B(z)}$
we have:
\begin{proposition}[Ref.~\citen{MN97}, Remark 2.2.2: Ref.~\citen{LZ95}, Proposition 2.3]\label{th:july1_1}
For any series $A(y)$ and $B(z)$
the following relation holds.
\begin{equation}
A(y)B(z) = \sum_{i=0}^{\infty} A(z)_{(i)}B(z) \, (y-z)^{-i-1}\big|_{|y|>|z|}
+\normOrd{A(y)B(z)} \label{eq:july1_2}
\end{equation}
\end{proposition}
This is an identity in
$({\rm End}\, M) [[y,y^{-1},z,z^{-1}]]$.
Based on this expression it is conventional to
introduce the following:
\begin{definition}\label{def:contraction}
For any series $A(y)$ and $B(z)$ their \textbf{contraction} is
\begin{equation}
\bcontraction{}{A}{(y)}{B} A(y)B(z) =
\contraction{}{A}{(y)}{B} A(y)B(z) =  \sum_{i=0}^{\infty} \frac{A(z)_{(i)}B(z)}{(y-z)^{i+1}}.
\label{eq:july3_2}
\end{equation}
\end{definition}
It is
an element of  $({\rm End}\, M) [[z,z^{-1}]][[(y-z)^{-1}]]$.
In Refs.~\citen{Kac96,MN97}
the correspondence $A(y)B(z)  \sim \contraction{}{A}{(y)}{B} A(y)B(z)$
is called the OPE.
\red{However, we use the following one instead:}
\begin{definition}\label{def:OPE}
Given two series $A(y)$ and $B(z)$,
their \textbf{OPE} is the sum of their contraction and their normally ordered product, which we denote by
\begin{equation}\label{eq:july29_3}
R(A(y)B(z)) := \contraction{}{A}{(y)}{B} A(y)B(z) + \normOrd{A(y)B(z)}.
\end{equation}
\end{definition}
It is an element of  $({\rm End}\, M) [[y,y^{-1},z,z^{-1}]][[(y-z)^{-1}]]$.
The notation $R$ stands for the radial ordering in the case of mutually local fields.
See Remark \ref{rem:mar29_8} below.
\subsection{Locality}\label{sec:2_3}
We now review the notion of locality \cite{MN97}.
This notion is important for us because the Wick formula \eqref{eq:main} is
proved to be equivalent to a special case of the Borcherds identity with the assumption of locality (Theorem \ref{th:2}).
Let $A(y), B(z)$ be elements of  $({\rm End}\, M) [[y,y^{-1}]]$ and
 $({\rm End}\, M) [[z,z^{-1}]]$ respectively.
Then their products $A(y)B(z)$ and $B(z)A(y)$ are
\red{not generally equal in} the set $({\rm End}\, M) [[y,y^{-1},z,z^{-1}]]$.
\red{Even in this case, it may be possible that} there exists a non-negative integer $m$ and the relation
\begin{equation}
A(y)B(z)(y-z)^m = B(z)A(y) (y-z)^m\label{eq:july6_3}
\end{equation}
is satisfied.
Then $A(z)$ and $B(z)$ are called mutually local.
If $m$ is the smallest non-negative integer such that the equality \eqref{eq:july6_3} holds,
then we say that the order of the locality of $A(z)$ and $B(z)$ is $m$.
By using the notation of the commutator $[A(y), B(z)]:=A(y)B(z)-B(z)A(y)$,
the condition is obviously written as $[A(y), B(z)](y-z)^m = 0$.
It should be emphasized that this condition does not necessarily implies
$[A(y), B(z)] = 0$ because $A(y), B(z)$ are not analytic functions
but just formal series.
However we have the following:
\begin{lemma}[Ref.~\citen{MN97}, Lemma 1.1.1.]\label{lem:mar17_5}
Let $a(y,z)$ be a series with only finitely many terms of negative or positive degree in $y$ or $z$.
If $a(y,z)$ satisfies $a(y,z) (y-z)^m =0$ for some nonnegative integer $m$,
then $a(y,z)=0$.
\end{lemma}

The following proposition on the OPE with the assumption of locality
will be used in the arguments in Sect.~\ref{sec:4}.
\begin{proposition}[Ref.~\citen{Kac96}, Theorem 2.3, Ref.~\citen{MN97}, Theorem 2.2.1]\label{th:mar15_1}
Let $A(y)$ and $B(z)$ be series on a vector space. 
They are mutually local if and only if both
\begin{equation}
A(y)B(z) = \sum_{i=0}^{m-1} A(z)_{(i)}B(z) \, (y-z)^{-i-1}\big|_{|y|>|z|}
+\normOrd{A(y)B(z)},  \label{eq:mar29_5} 
\end{equation}
and 
\begin{equation}
B(z)A(y) = \sum_{i=0}^{m-1} A(z)_{(i)}B(z) \, (y-z)^{-i-1}\big|_{|y|<|z|}
+\normOrd{A(y)B(z)}, \label{eq:mar29_6}
\end{equation}
hold for some $m \in \mathbb{N}$.
\end{proposition}
Although proofs are available in the original papers\cite{Kac96, MN97},
we present another heuristic one here under the assumption of Proposition \ref{th:july1_1}.
\par\noindent
\proof
If we multiply $(y-z)^m$ to \eqref{eq:mar29_5} and \eqref{eq:mar29_6},
then their right hand sides are equal to each other
because we have $(y-z)^m (y-z)^{-i-1}\big|_{|y|>|z|}= (y-z)^m(y-z)^{-i-1}\big|_{|y|<|z|}$.
Hence if these two equations are satisfied, then $A(z)$ and $B(z)$ are mutually local.
Conversely, suppose these fields are mutually local,
and let $m$ be the order of locality.
Then by the definition of the residue product \eqref{eq:oct20_1}
we have $A(z)_{(i)}B(z) = 0$ for $i \geq m$.
Hence \eqref{eq:mar29_5} follows from Proposition \ref{th:july1_1}.
On the other hand, by the assumption of locality $[A(y), B(z)](y-z)^m = 
[A(y)_+ + A(y)_-, B(z)] (y-z)^m=0$ 
we obtain the following relation:
\begin{align}
&\left\{\hphantom{-}
[A(y)_+, B(z)] - \sum_{i=0}^{m-1} A(z)_{(i)}B(z) \, (y-z)^{-i-1}\big|_{|y|>|z|}
\right\}(y-z)^m \nonumber\\
=&\left\{
-[A(y)_-, B(z)] - \sum_{i=0}^{m-1} A(z)_{(i)}B(z) \, (y-z)^{-i-1}\big|_{|y|<|z|}
\right\}(y-z)^m. \label{eq:mar29_7}
\end{align}
Since \eqref{eq:mar29_5} has already been proved, 
one sees that the left hand side of this expression vanishes.
Hence one sees that \eqref{eq:mar29_6} follows by using
Lemma \ref{lem:mar17_5},
because
there is no term of negative degree in $y$ in the expression
between the braces in the right hand side of \eqref{eq:mar29_7},
and by noting that $-[ A(y)_-, B(z)]$ is equal to the difference of $B(z)A(y)$
and $\normOrd{A(y)B(z)}$.
\qed

\begin{remark}\label{rem:mar29_8}
In this section the operators are defined as
formal series and 
we are not allowed to
substitute any real or complex numbers into
the symbols $y$ and $z$.
However, in Sect.~\ref{sec:4} we will introduce an analytic formulation
in which these variables can take values in complex numbers.
Then, from Proposition \ref{th:mar15_1} we see that
the symbol $R$ in Definition \ref{def:OPE} 
can be interpreted as the {radial ordering} of the operators.
This means that in any correlation functions involving mutually local
operators $A(y),B(z)$ we can replace
$R(A(y)B(z))$ by $A(y)B(z)$ (resp.~$B(z)A(y)$) if
the condition $|y|>|z|$ (resp.~$|y|<|z|$) is satisfied.
\end{remark}

We now present the notion of locality between more than two fields.
Say series $A^{1}(z),\dots, A^{\ell}(z)$ are mutually local if
all the distinct pairs $A^{i}(z),  A^{j}(z), (i \ne j)$ are mutually local.
The following proposition is known as the Dong's lemma \cite{Kac96}.
\begin{proposition}[Ref.~\citen{MN97}, Proposition 2.1.5.]\label{prop:mar17_3}
If $A(z), B(z)$ and $C(z)$ are mutually local fields, then $A(z)_{(m)}B(z)$ and $C(z)$ are
local.
\end{proposition}
This and the following results will be used in Sect.~\ref{sec:4}.
\begin{proposition}[Ref.~\citen{MN97}, Proposition 2.1.6.]\label{prop:mar17_4}
Let $A^{i}(z) = \sum_{n \in \mathbb{Z}} A^{i}_n z^{-n-1}$ be mutually local fields 
for $1 \leq i  \leq \ell$.
Then for any $u \in M$ the condition
\begin{equation*}
A^{1}_{p_1} \cdots A^{\ell}_{p_\ell} u = 0,
\end{equation*}
is satisfied for any sufficiently large $p_1 + \cdots + p_\ell \in \Z$.
\end{proposition}
\subsection{Differentiation and Taylor expansion}\label{sec:2_4}
We now introduce some formulas
on the differentiation of fields.
For the series $A(z)$ defined as \eqref{eq:july6_1} we define
its differentiation $\partial_z A(z)$ as
\begin{equation}
\partial_z A(z) = \sum_{n \in \mathbb{Z}} (-n-1) A_n z^{-n-2}.\label{eq:july6_2}
\end{equation}
If $A(z)$ is a field, so is $\partial_z A(z)$.
Also, the differentiation of a field of two variables is given in a similar way.
From \eqref{eq:oct20_1} and \eqref{eq:july6_2} one easily proves the following formula
\begin{equation}
(\partial_z A(z))_{(i)} B(z) = -i A(z)_{(i-1)}B(z).\label{eq:july29_2}
\end{equation}
By repeated use of \eqref{eq:july29_2} we obtain
\begin{equation}
\normOrd{\partial^{(j)}A(z) B(z)} = A(z)_{(-j-1)}B(z) ,\label{eq:mar16_1}
\end{equation}
for any non-negative integer $j$.
Let $I(z)$ be the identity field on $M$: That is an operator whose
only non-zero term is the constant term which is the identity operator on $M$.
Then by \eqref{eq:mar16_1} we have:
\begin{proposition}[Ref.~\citen{MN97}, Proposition 1.4.4.]\label{prop:mar16_2}
\begin{equation}
A(z)_{(m)} I(z) =
\begin{cases}
0, & (m \geq 0), \\
\partial^{(-m-1)} A(z), & (m \leq -1). 
\end{cases}
\end{equation}
\end{proposition}
This relation will be used to derive the skew symmetry from the Borcherds identity. Also, we have the following Taylor's formula
which will be used in the proof of Theorem \ref{th:2}, 
as well as in some calculations in Appendix \ref{sec:5}.
\begin{proposition}[Ref.~\citen{MN97}, Proposition 2.2.4.]\label{prop:jul30_3}
For any field 
$A(y,z)$ and any positive integer
$N$, there exists a unique field
$R_N(y,z)$ such that the following relation is satisfied:
\begin{equation*}
A(y,z) = \sum_{i=0}^{N-1} \partial_y^{(i)} A(y,z)\vert_{y=z} (y-z)^i + R_N(y,z) (y-z)^N.
\end{equation*}
In particular for any field
$A(z)$ and any positive integer
$N$, there exists a unique field
$R_N(y,z)$ such that the following relation is satisfied:
\begin{displaymath}
A(y) = \sum_{i=0}^{N-1} \partial_z^{(i)} A(z) (y-z)^i + R_N(y,z) (y-z)^N.
\end{displaymath}
\end{proposition}
Indeed the latter statement follows from the former, since
by setting $A_{p,-1}=A_p, A_{p,q(\ne -1)}=0$ in \eqref{eq:sep4_2}
any field of one variable can be obtained as a special case of
a field of two variables.

\subsection{Borcherds identity and skew symmetry}\label{sec:2_5}
Now we introduce a remarkable identity that is relevant for our discussion
on the Wick theorems.
As one of the main results of Matsuo and Nagatomo \cite{MN97}, 
they derived an identity that comprises
three infinite sums of nested commutators
of arbitrary three fields with some binomial expansion factors,
which is a consequence of the usual Jacobi identity.
By taking the residue on both sides of this identity they 
proved the following identity satisfied by arbitrary three fields
with respect to the residue products.
\begin{proposition}[Ref.~\citen{MN97}, Corollaries 3.2.2. and 3.4.2.]\label{th:Borcherds}
Let $A(w), B(w)$ and $C(w)$ be fields on a vector space.
Then, for any $p,r \in \mathbb{N}$ and any $q \in \Z$,
\begin{eqnarray}
&&\sum_{i=0}^\infty {p \choose i} (A(w)_{(r+i)}B(w))_{(p+q-i)}C(w)\nonumber\\
&&=\sum_{i=0}^\infty (-1)^i  {r \choose i}
\left( A(w)_{(p+r-i)}(B(w)_{(q+i)} C(w)) 
-(-1)^r
B(w)_{(q+r-i)} (A(w)_{(p+i)} C(w)) \right).\nonumber\\
&&\label{eq:borcherds}
\end{eqnarray}
Moreover, if $A(w), B(w), C(w)$ are mutually local, then this relation is satisfied
for any $p,q,r \in \Z$.  
\end{proposition}
This relation \eqref{eq:borcherds} is called the Borcherds identity \cite{Kac96},
since it is equivalent to one of the axioms for a vertex algebra 
introduced by Borcherds \cite{B86}.

We also introduce an identity that will be used in Appendix~\ref{sec:5}.
As a special case of the Borcherds identity, one can derive the following relation
which is known as the skew symmetry.
\begin{proposition}[Ref.~\citen{MN97}, Proposition 3.5.2.]\label{pr:skewsym}
\begin{equation}\label{eq:skewsym}
B(z)_{(m)}A(z) = \sum_{i=0}^\infty (-1)^{m+i+1} \partial^{(i)}
(A(z)_{(m+i)}B(z) ).
\end{equation}
\end{proposition}
Indeed this relation is obtained
by setting $p=-1, q=0, r=m, B(z) = I(z)$ in \eqref{eq:borcherds}
and by using Proposition \ref{prop:mar16_2}.

\section{Wick Theorems as Special Cases of the Borcherds Identity}\label{sec:3}
\subsection{Generalized Wick theorem for $\contraction{}{A}{(z)(}{BC}A(z)(BC)(w)$}\label{sec:3_1}
In this section we show that the Wick theorems \eqref{eq:july3_1} and \eqref{eq:main} are equivalent to special cases of the Borcherds identity.
For this purpose we first prepare some notations.
By generalizing our notation
in Definition \ref{def:OPE}
we define 
the OPEs of
an operator and a contraction as
\begin{eqnarray}
R(\contraction{}{A}{(z)}{B}A(z)B(x)C(w)) &=&  \contraction{(}{A}{(z)}{B}
\bcontraction{(A(z}{)}{B(x))}{C}
(A(z)B(x))C(w) +\normOrdx{\,\contraction{}{A}{(z)}{B} A(z)B(x)C(w)},\label{eq:jan27_2a}\\
R(B(x) \contraction{}{A}{(z)}{C}A(z) C(w) ) &=& \bcontraction{}{B}{(x) (A(z}{)}
\contraction{B(x) (}{A}{(z)}{C}
B(x) (A(z) C(w)) + \normOrdx{\,B(x) \contraction{}{A}{(z)}{C} A(z)C(w)} ,
\label{eq:jan27_2}
\end{eqnarray}
where
\begin{equation}
\contraction{(}{A}{(z)}{B}
\bcontraction{(A(z}{)}{B(x))}{C}
(A(z)B(x))C(w)=
\sum_{i=0}^{\infty} \frac{
\contraction{(A(x)}{_(i)}{B(x))}{C}(A(x)_{(i)}B(x))C(w)}{(z-x)^{i+1}}=
\sum_{i=0}^{\infty} \sum_{j=0}^{\infty}\frac{(A(w)_{(i)}B(w))_{(j)}C(w)}{(z-x)^{i+1}(x-w)^{j+1}},
\end{equation}
\begin{equation}
\normOrdx{\,\contraction{}{A}{(z)}{B} A(z)B(x)C(w)}=
\sum_{i=0}^{\infty} \frac{\normOrd{ \,(A(x)_{(i)}B(x)) C(w)}}{(z-x)^{i+1}},
\end{equation}
\begin{equation}\label{eq:mar24_5}
\bcontraction{}{B}{(x) (A(z}{)}
\contraction{B(x) (}{A}{(z)}{C}
B(x) (A(z) C(w))=
\sum_{i=0}^{\infty} \frac{\contraction{}{B}{(x)(A(w)_}{(i}B(x) (A(w)_{(i)}C(w))}{(z-w)^{i+1}}
=
\sum_{i=0}^{\infty} 
\sum_{j=0}^{\infty} 
\frac{B(w)_{(j)} (A(w)_{(i)}C(w))}{(x-w)^{j+1} (z-w)^{i+1}},
\end{equation}
\begin{equation}\label{eq:jan27_1}
\normOrdx{\,B(x) \contraction{}{A}{(z)}{C} A(z)C(w)}=
\sum_{i=0}^{\infty} \frac{\normOrd{ \,B(x) (A(w)_{(i)}C(w))}}{(z-w)^{i+1}}.
\end{equation}

As we have mentioned before, 
the operators are defined as formal series and we are not allowed to
substitute any real or complex numbers into the symbols $w,x$ and $z$.
Admitting that, from now on we will try to write formulas including 
contour integrals of
the operators with respect to these variables. 
Actually, they are simply interpreted as formal symbols
for an algebraic manipulation that uses the following formula:
\begin{equation}
\frac{1}{2 \pi \sqrt{-1}}
\oint_{C_w} \frac{d x}{(z-x)^m (x-w)^n} 
= {m+n-2 \choose n-1}
\frac{1}{(z-w)^{m+n-1}}. \label{eq:sep4_4}
\end{equation}

Now we present the well-known generalized Wick theorem in our formulation:
\begin{theorem}\label{th:1}
\red{Let $A(w), B(w)$ and $C(w)$ be fields on a vector space.
Then the following relation is satisfied.}
\begin{equation}
\contraction{}{A}{(z)(}{BC}
A(z)(BC)(w) 
= \frac{1}{2 \pi \sqrt{-1}}
\oint_{C_w} \frac{d x}{x-w} \left\{R(\contraction{}{A}{(z)}{B}
A(z)B(x) C(w))+ R(B(x) \contraction{}{A}{(z)}{C}A(z) C(w) )\right\}.
\label{eq:oct21_1}
\end{equation}
Here $R$ represents the OPE given by \eqref{eq:jan27_2a}
and \eqref{eq:jan27_2}.
\end{theorem}
\begin{proof}
By using the formula \eqref{eq:sep4_4}
one can rewrite \eqref{eq:oct21_1} as
\begin{eqnarray*}
\sum_{p=0}^\infty \frac{A(w)_{(p)} (B(w)_{(-1)}C(w))}{(z-w)^{p+1}} &=&
\sum_{i=0}^\infty \sum_{j=0}^\infty { i+j \choose i }
\frac{  (A(w)_{(i)} B(w))_{(j-1)} C(w)      }{(z-w)^{i+j+1}   }\\
&&+
\sum_{p=0}^\infty \frac{B(w)_{(-1)} ( A(w)_{(p)} C(w)  )  }{(z-w)^{p+1}  }.
\end{eqnarray*}
Note that this is {a ``generating identity'' for the following identities}
\begin{equation}\label{eq:mar31_1}
A(w)_{(p)} (B(w)_{(-1)}C(w)) =
\sum_{i=0}^p  {p \choose i }
 (A(w)_{(i)} B(w))_{(p-i-1)} C(w) +
B(w)_{(-1)} ( A(w)_{(p)} C(w)  ) ,
\end{equation}
for any $p \in \mathbb{N}$.
Now one finds that this is \red{just} a special case of the 
Borcherds identity for non-local fields given by
Proposition \ref{th:Borcherds}.
In fact, by setting $r=0$ and $q=-1$ in \eqref{eq:borcherds}, we obtain the above result.
\qed
\end{proof}

\begin{remark}
In Ref.~\citen{Kac96} Kac presented a proof of this
special case of the Borcherds identity.
He called it the {\em non-commutative Wick formula} and 
pointed out its equivalence to the formula in Ref.~\citen{BBSS88}.
He also presented an integral formula \eqref{eq:mar29_1x},
that is another type of generating identity for \eqref{eq:mar31_1}.
\end{remark}

\subsection{Generalized Wick theorem for $\contraction{}{(AB)}{(z)(}{C}(AB)(z)C(w)$}\label{sec:3_2}
Now we present the main result of this paper in algebraic formulation.
\begin{theorem}\label{th:2}
\red{Let $A(w), B(w)$ and $C(w)$ be mutually local fields on a vector space.
Then the following relation is satisfied.}
\begin{equation}\label{eq:main2}
\contraction{}{(AB)}{(z)(}{C}
(AB)(z)C(w) 
= \frac{1}{2 \pi \sqrt{-1}}
\oint_{C_w} \frac{d x}{z-x} \left\{
\normOrdx{A(x) \contraction{}{B}{(x)}{C}B(x)C(w)} + R(B(x) \contraction{}{A}{(x)}{C}A(x) C(w)) \right\}.
\end{equation}
Here $R$ represents the OPE given by \eqref{eq:jan27_2}.
\end{theorem}
\proof
From \eqref{eq:jan27_1} we have
\begin{eqnarray}\label{eq:apr4_1}
\normOrdx{\,A(x) \contraction{}{B}{(x)}{C} B(x)C(w)}&=&
\sum_{i=0}^{\infty} \frac{\normOrd{ \,A(x) (B(w)_{(i)}C(w))}}{(x-w)^{i+1}} \nonumber\\
&=& 
\sum_{i=0}^{\infty} \sum_{j=0}^{\infty}
\frac{A(w)_{(-j-1)} (B(w)_{(i)}C(w))}{(x-w)^{i-j+1}}, 
\end{eqnarray}
where we have used Taylor's formula (Proposition \ref{prop:jul30_3}) for $A(x)$ 
\red{at the point $x=w$}, and then used the formula \eqref{eq:mar16_1}.
For the sake of simplicity we formally expanded $A(x)$ up to infinite order,
but this does not cause any problems because the regular terms will be dropped
by integration.
Actually by using \eqref{eq:sep4_4} we obtain
\begin{eqnarray}\label{eq:jan27_5}
&&\frac{1}{2 \pi \sqrt{-1}}
\oint_{C_w} \frac{d x}{z-x} \left\{
\normOrdx{A(x) \contraction{}{B}{(x)}{C}B(x)C(w)} \right\}
\nonumber\\
&&=
\sum_{i=0}^{\infty} \sum_{j=0}^{i}
\frac{A(w)_{(-j-1)} (B(w)_{(i)}C(w))}{(z-w)^{i-j+1}}\nonumber\\
&&=
\sum_{q=0}^{\infty} \sum_{j=0}^{\infty}
\frac{A(w)_{(-j-1)} (B(w)_{(q+j)}C(w))}{(z-w)^{q+1}}.
\end{eqnarray}
From \eqref{eq:jan27_2} we have
\begin{equation}
R(B(x) \contraction{}{A}{(x)}{C}A(x) C(w) ) = \bcontraction{}{B}{(x) (A(x}{)}
\contraction{B(x) (}{A}{(x)}{C}
B(x) (A(x) C(w)) + \normOrdx{\,B(x) \contraction{}{A}{(x)}{C} A(x)C(w)},
\end{equation}
where
\begin{equation}
\bcontraction{}{B}{(x) (A(x}{)}
\contraction{B(x) (}{A}{(x)}{C}
B(x) (A(x) C(w))=
\sum_{i=0}^{\infty} 
\sum_{j=0}^{\infty} 
\frac{B(w)_{(j)} (A(w)_{(i)}C(w))}{(x-w)^{i+j+2} }.
\end{equation}
Hence by using \eqref{eq:sep4_4} we obtain
\begin{eqnarray}
&&\frac{1}{2 \pi \sqrt{-1}}
\oint_{C_w} \frac{d x}{z-x} \left\{
R(B(x) \contraction{}{A}{(x)}{C}A(x) C(w) )
\right\}\nonumber\\
&&=
\sum_{i=0}^{\infty} 
\sum_{j=0}^{\infty} 
\frac{B(w)_{(j)} (A(w)_{(i)}C(w))}{(z-w)^{i+j+2} }
+
\sum_{i=0}^{\infty} \sum_{j=0}^{i}
\frac{B(w)_{(-j-1)} (A(w)_{(i)}C(w))}{(z-w)^{i-j+1}}\nonumber\\
&&=
\sum_{i=0}^{\infty} 
\sum_{j=-i-1}^{\infty} 
\frac{B(w)_{(j)} (A(w)_{(i)}C(w))}{(z-w)^{i+j+2} }
\nonumber\\
&&=
\sum_{q=0}^{\infty} 
\sum_{i=0}^{\infty} 
\frac{B(w)_{(q-i-1)} (A(w)_{(i)}C(w))}{(z-w)^{q+1} }.
\label{eq:jan27_6}
\end{eqnarray}
\red{Here we used \eqref{eq:jan27_5} for the first equality to derive the second term.}
From \eqref{eq:jan27_5} and \eqref{eq:jan27_6} we see that \red{the integral formula}
\eqref{eq:main2}
{is a ``generating identity'' for the following identities}
\begin{align}\label{eq:mar31_2}
(A(w)_{(-1)}B(w))_{(q)}C(w)=
\sum_{i=0}^{\infty}
&\left( A(w)_{(-i-1)} (B(w)_{(q+i)}C(w))
+B(w)_{(q-i-1)} (A(w)_{(i)}C(w)) \right),
\end{align}
\red{for any $q \in \mathbb{N}$.}
Now one finds that this is \red{just} a special case of the 
Borcherds identity for local fields given by 
Proposition \ref{th:Borcherds}.
In fact, by setting $p=0$ and $r=-1$ in \eqref{eq:borcherds}, we obtain the above \red{identity}.
\qed
\begin{remark}
By using \eqref{eq:mar16_1}
one can rewrite \eqref{eq:mar31_2} into the following form
\begin{align}\label{eq:oct27x_1}
(A(w)_{(-1)}B(w))_{(q)}C(w)&=
\sum_{i=0}^{\infty}
\left( \normOrd{\partial^{(i)}A(w) \, (B(w)_{(q+i)}C(w))}
+\normOrd{\partial^{(i)}B(w) \, (A(w)_{(q+i)}C(w))} \right)
\nonumber\\
&+\sum_{i=1}^q B(w)_{(i-1)} (A(w)_{(q-i)}C(w)).
\end{align}
We note that this series of identities itself was known for a long time \cite{Th95}.
\end{remark}
\section{Derivation of the Formulas by an Analytic Method}\label{sec:4}
\subsection{Admissible fields and the restricted dual space}\label{sec:4_1}
In this section we derive the generalized Wick formulas
by an analytic method.
Our argument is based on the formulation used in Appendix B of Ref.~\citen{MN97}.
\red{We shall review it in this subsection somehow in detail
to make our arguments in the subsequent subsections clearer.}
Besides the locality, we assume one additional condition of the fields
which is called the admissibility.

Let $M$ be a $\C$-vector space and $M^*$ the dual of $M$.
We denote the canonical paring by 
\begin{equation*}
\langle \quad , \quad \rangle : M^* \times M \rightarrow \C.
\end{equation*} 
In other words, $M^*$ is the set of all linear functions on $M$
which is also a $\C$-vector space,
and we write $\varphi (v) = \langle \varphi , v \rangle$ for
$v \in M, \varphi \in M^*$.
\red{Let $N$ be a subspace of $M$ and set
\begin{equation*}
N^{\perp} = \{ v^{\perp} \in M^*
\vert \langle v^{\perp} , N \rangle = 0 \}, \quad
(N^{\perp})^{\perp} = \{ v \in M
\vert \langle N^{\perp} , v \rangle = 0 \}.
\end{equation*}
Then we have $N \subset (N^{\perp})^{\perp}$ but the opposite inclusion
is not generally satisfied if $M$ is an infinite dimensional vector space.
However, in this paper we simply assume that the condition 
$N = (N^{\perp})^{\perp}$
is satisfied for any subspace of $M$.}
A subspace $M^{\vee}  \subset M^*$ is called \textbf{nondegenerate} if
there is no $u \in M$ but $u=0$ that satisfies
the condition $\langle M^{\vee} , u \rangle = 0$.
Then:
\begin{lemma}[Ref.~\citen{MN97}, Lemma B.1.1]\label{lem:mar17_2}
Let $N_m \,(m \in \Z)$ be subspaces of $M$ such that $\cdots \subset N_m \subset N_{m+1} \subset \cdots$ and $\bigcap_{m \in \Z} N_m = \{ 0 \}$.
Then there exists a nondegenerate subspace $M^{\vee} \subset M^*$ 
such that for any $v^{\vee} \in M^{\vee}$, there exists an $m \in \Z$
such that $\langle v^{\vee}, N_m \rangle = 0$.
\end{lemma}
A proof is available in Ref.~\citen{MN97}.
The following one is essentially the same one but
several explanations are supplemented.
\par\noindent
\proof
Set $N = \bigcup_{m \in \Z} N_m$ and consider 
\begin{equation*}
N_m^{\perp} = \{ v^{\perp} \in N^*
\vert \langle v^{\perp} , N_m \rangle = 0 \} \subset N^*.
\end{equation*}
Then $N_m^{\perp}\, (m \in \Z)$ are subspaces of $N^*$ such that
$\cdots \subset N_{m+1}^{\perp} \subset N_{m}^{\perp} \subset \cdots$.
Take a complement $P$ of $N$ in $M$ so that $M = N \oplus P$, 
that implies $N \cap P = \{ 0 \}$.
Since $\left(\bigcup_{m \in \Z}  N_{m}^{\perp} \right)$ and $P^*$ are subspaces
of $M^*$, so is their sum space which is defined as
\begin{equation*}
\left(\bigcup_{m \in \Z}  N_{m}^{\perp} \right) + P^* =
\left\{ v_1 + v_2 \in M^* \Bigg\vert v_1 \in \left(\bigcup_{m \in \Z}  N_{m}^{\perp}\right),
v_2 \in P^* \right\}.
\end{equation*}
Since $N \cap P = \{ 0 \}$ implies $N^* \cap P^* = \{ 0 \}$ and $\left(\bigcup_{m \in \Z}  N_{m}^{\perp} \right) \subset N^*$,
this sum is actually a direct sum which we denote by
\begin{equation*}
M^{\vee} = \left(\bigcup_{m \in \Z}  N_{m}^{\perp} \right) \oplus P^* \subset M^*. 
\end{equation*}
Then this $M^{\vee}$ has the desired properties
as one sees in the followings.

Given any $v^{\vee} \in M^{\vee}$, it can be uniquely written as
$v^{\vee} = v_1 + v_2, \, v_1 \in \left(\bigcup_{m \in \Z}  N_{m}^{\perp} \right)$
and $v_2 \in P^*$.
The former implies $\langle v_1, N_m \rangle = 0$ for some $m \in \Z$, and
the latter implies $\langle v_2, N_m \rangle = 0$ for any $m \in \Z$ because $N \cap P = \{ 0 \}$.
Hence for any $v^{\vee} \in M^{\vee}$, there exists an $m \in \Z$
such that $\langle v^{\vee}, N_m \rangle = 0$.

In order to prove the nondegeneracy of $M^{\vee}$, suppose $\langle M^{\vee}, u \rangle = 0$.
Recall that any $u \in M= N \oplus P$ can be uniquely written as
$u = u_1 + u_2, \, u_1 \in N$
and $u_2 \in P$.
The nondegeneracy of $P^*$ with respect to $P$ requires $u_2 = 0$.
Now $\langle M^{\vee}, u_1 \rangle = 0$ requires
$u_1$ to be an element of $\left( N_m^{\perp} \right)^{\perp}$ for all $m \in \Z$ where
\begin{equation*}
\left( N_m^{\perp} \right)^{\perp} = \{ v \in N
\vert \langle N_m^{\perp} , v \rangle = 0 \} \subset N.
\end{equation*}
Then, since $\left( N_m^{\perp} \right)^{\perp}  = N_m$
and $\bigcap_{m \in \Z} N_m = \{ 0 \}$,
we have $u_1 =0$.
\hfill
\qed

Now we explain the notion of admissibility.
Let
\begin{equation}\label{eq:mar17_1}
A^{i}(z) = \sum_{n \in \mathbb{Z}} A^{i}_n z^{-n-1} \qquad
(1 \leq i  \leq \ell)
\end{equation}
be formal series in $({\rm End}\, M) [[z,z^{-1}]]$.
We denote by $S_\ell$ the symmetric group acting on $\{ 1, \ldots, \ell \}$.
For $u \in M, m \in \Z$ we define
\begin{equation*}
N_{u,m} = \sum_{p_1+\cdots+p_\ell = - \infty}^m \sum_{\sigma \in S_\ell}
\C A_{p_1}^{\sigma (1)} \cdots A_{p_\ell}^{\sigma (\ell)} u.
\end{equation*}
Here $\C A_{p_1}^{\sigma (1)} \cdots A_{p_\ell}^{\sigma (\ell)} u$
is the one-dimensional subspace of $M$,
and as a sum of the subspaces of a vector space, so is $N_{u,m}$.
Clearly they satisfy the condition $\cdots \subset N_{u,m} \subset N_{u,m+1} \subset \cdots$.
We say that the series \eqref{eq:mar17_1} are \textbf{admissible}
if the condition $\cap_{m \in \Z} N_{u,m}= \{ 0 \}$ is satisfied for any $u \in M$.

Suppose that the series \eqref{eq:mar17_1} are admissible.
Then from Lemma \ref{lem:mar17_2} we observe that
for any $u \in M$ there exists a nondegenerate subspace $M^{\vee}_u \subset M^*$ such that for any $v^{\vee} \in M^{\vee}_u, \sigma \in S_\ell$
the condition
\begin{equation}\label{eq:july28_1}
\langle v^{\vee} , A_{p_1}^{\sigma (1)} \cdots A_{p_\ell}^{\sigma (\ell)} u \rangle =0,
\end{equation}
is satisfied for any sufficiently small $p_1 + \cdots + p_\ell \in \Z$.
Such an $M^{\vee}_u \subset M^*$ is called a \textbf{restricted dual space} compatible with
$A^{1}(z),\dots, A^{\ell}(z)$ with respect to $u \in M$.

\subsection{Consequences of locality and admissibility}\label{sec:4_2}
In order to derive the generalized Wick formulas analytically,
we have to require the matrix elements
\begin{equation*}
\langle v^{\vee} ,  A^{1}(z_1) \cdots A^{\ell}(z_\ell) u \rangle,
\end{equation*}
which we call  the correlation functions,
to satisfy a certain property with respect to the 
permutations of the indices $1,\dots,\ell$.
Actually only the cases $\ell = 2,3$ are relevant for our task.
For this purpose we introduce:
%
\begin{proposition}\label{prop:mar18_1}
Suppose $\ell = 2$ or $3$ and
let $A^{1}(z),\dots, A^{\ell}(z)$ be admissible fields.
If they are local, then for any $u \in M, v^{\vee} \in M^{\vee}_u$, and
$\sigma \in S_\ell$, the correlation functions
\begin{equation*}
\langle v^{\vee} ,  A^{\sigma(1)}(z_{\sigma(1)}) \cdots A^{\sigma(\ell)}(z_{\sigma(\ell)}) u \rangle
\end{equation*}
are the expansions of a rational function of the form
\begin{equation*}
\frac{P(z_1,\dots,z_\ell)}{\prod_{i<j}(z_i-z_j)^{n_{ij}}},\quad
P(z_1,\dots,z_\ell) \in \C [z_1,z_1^{-1},\dots,z_\ell,z_\ell^{-1}],
\end{equation*}
that is common to all $\sigma \in S_\ell$,
into the regions satisfying $|z_{\sigma(1)}| > \dots > |z_{\sigma(\ell)}|$.
\end{proposition}
\proof
\red{Since $A^{\sigma(\ell)}(z_{\sigma(\ell)})$ is a field,
the correlation function $\langle v^{\vee} ,  A^{\sigma(1)}(z_{\sigma(1)}) \cdots A^{\sigma(\ell)}(z_{\sigma(\ell)}) u \rangle$
has only finitely many terms of negative degree in $z_{\sigma(\ell)}$.
It also has only finitely many terms of positive degree in $z_{\sigma(1)}$.
This is a consequence of the admissibility which implies \eqref{eq:july28_1}, and 
the locality which implies Proposition \ref{prop:mar17_4}.}
Now, take sufficiently large $n_{ij} \in \N, (i<j)$, and consider the series
\begin{equation*}
\langle v^{\vee} ,  A^{\sigma(1)}(z_{\sigma(1)}) \cdots A^{\sigma(\ell)}(z_{\sigma(\ell)}) u \rangle \prod_{i<j}(z_i-z_j)^{n_{ij}}.
\end{equation*}
Then by the locality which enables us to move any operators to the leftmost or the rightmost position and the above observation, one sees that they are equal to a Laurent polynomial $P(z_1,\dots,z_\ell)$
that is common to all $\sigma \in S_\ell$.

As an example we consider the case with
$\ell = 3$ and $\sigma = {\rm Id}$.
Let
\begin{align*}
F_0(z_1, z_2, z_3) &= \langle v^{\vee} ,  A^{1}(z_1) A^{2}(z_2) A^{3}(z_3)u \rangle
-\frac{P(z_1,z_2,z_3)}{\prod_{i<j}(z_i-z_j)^{n_{ij}}}\bigg|_{|z_1|>|z_2|>|z_3|},\\
F_1(z_1, z_2, z_3) &=F_0(z_1, z_2, z_3) (z_1-z_3)^{n_{13}},\\
F_2(z_1, z_2, z_3) &=F_1(z_1, z_2, z_3) (z_1-z_2)^{n_{12}},\\
F_3(z_1, z_2, z_3) &=F_2(z_1, z_2, z_3) (z_2-z_3)^{n_{23}}.
\end{align*}
From the above arguments
one sees that $F_3(z_1, z_2, z_3) =0$,
and that $F_i(z_1, z_2, z_3)$ has
only finitely many terms of positive degree in $z_{1}$ and negative degree in $z_{3}$
for $(0\leq i \leq 2)$.
Hence by repeated use of
Lemma \ref{lem:mar17_5} we obtain the desired result $F_0(z_1, z_2, z_3) =0$.
The other cases are proved in a similar way.
\qed

\begin{remark}
This proposition and its proof is based on Theorem B.2.2 in Ref.~\citen{MN97}, which claims
an analogous result for arbitrary $\ell$.
For the soundness of our arguments we restrict ourselves to the above cases,
and supplemented several explanations to the proof.
\end{remark}

\subsection{Generalized Wick theorem for $\langle v^{\vee} , \contraction{}{A}{(z)(}{BC}A(z)(BC)(w)u \rangle$}\label{sec:4_3}
We now present the well-known formula \eqref{eq:july3_1} under the analytic
formulation introduced in Sect.~\ref{sec:4_1}.
Let $A(z), B(z), C(z)$ be mutually local and admissible fields on a vector space $M$, and
$M^{\vee}_u \subset M^*$ be the restricted dual space compatible with them
with respect to $u \in M$.
\begin{theorem}\label{th:leftwick}
For any $u \in M, v^{\vee} \in M^{\vee}_u$ the following relation holds:
\begin{align*}
&\langle v^{\vee} , \contraction{}{A}{(z)(}{BC}
A(z)(BC)(w)  u \rangle \\
&= \frac{1}{2 \pi \sqrt{-1}}
\oint_{C_w} \frac{d x}{x-w} \langle v^{\vee} , \left\{R(\contraction{}{A}{(z)}{B}
A(z)B(x) C(w))+ R(B(x) \contraction{}{A}{(z)}{C}A(z) C(w) )\right\}u \rangle .
\end{align*}
Here $R$ represents the OPE given by \eqref{eq:jan27_2a}
and \eqref{eq:jan27_2}.
\end{theorem}
\proof
By the definition of the residue product \eqref{eq:oct20_1},
the following relation is satisfied for any $p \in \N$:
\begin{align*}
&\langle v^{\vee} , A(w)_{(p)} (B(w)_{(-1)}C(w)) u \rangle\\
&={\rm Res}_{y=0} {\rm Res}_{x=0} 
\langle v^{\vee} , [A(y), [B(x), C(w)]] u \rangle (y-w)^p (x-w)^{-1}.
\end{align*}
Here the factor $ (x-w)^{-1}$ should be interpreted as its expansion
in the region $|x|>|w|$ (resp.~$|w|>|x|$) if the order of
the product of the operators is $B(x)C(w)$ (resp.~$C(w)B(x)$)
in each of the four terms in $[A(y), [B(x), C(w)]]$.
Let
\begin{align}
& 0 < R_1 < R_2 < |w| < R_3 < R_4,\nonumber\\
& C_{i, \zeta} = \{ z \in \C \mid |z-\zeta| = R_i \}, \, (i=1,\dots,4),\label{eq:mar23_4}
\end{align}
and $C_\zeta$ be a contour encircling the point $\zeta\in \C$ with an infinitesimal radius.
Then we have
\begin{align}
&{\rm Res}_{y=0} {\rm Res}_{x=0} 
\langle v^{\vee} , [A(y), [B(x), C(w)]] u \rangle (y-w)^p (x-w)^{-1}\nonumber\\
&=\frac{1}{(2 \pi \sqrt{-1})^2}
\oint_{C_{4,0}}\!\!\!d y \oint_{C_{3,0}}\!\!\!d x
\langle v^{\vee} , A(y)B(x)C(w) u \rangle(y-w)^p (x-w)^{-1}\nonumber\\
&-\frac{1}{(2 \pi \sqrt{-1})^2}
\oint_{C_{4,0}}\!\!\!d y \oint_{C_{2,0}}\!\!\!d x
\langle v^{\vee} , A(y)C(w)B(x) u \rangle(y-w)^p (x-w)^{-1}\nonumber\\
&-\frac{1}{(2 \pi \sqrt{-1})^2}
\oint_{C_{1,0}}\!\!\!d y \oint_{C_{3,0}}\!\!\!d x
\langle v^{\vee} , B(x)C(w)A(y) u \rangle(y-w)^p (x-w)^{-1}\nonumber\\
&+\frac{1}{(2 \pi \sqrt{-1})^2}
\oint_{C_{1,0}}\!\!\!d y \oint_{C_{2,0}}\!\!\!d x
\langle v^{\vee} , C(w)B(x)A(y) u \rangle(y-w)^p (x-w)^{-1}.\label{eq:mar23_1}
\end{align}
According to Proposition \ref{prop:mar18_1} we see that $\langle v^{\vee} , A(y)B(x)C(w) u \rangle$
and its permutations are expansions of the same rational function
which we denote by
\begin{equation}
\langle\langle v^{\vee} , A(y)B(x)C(w) u \rangle\rangle.\label{eq:mar23_2}
\end{equation}
For instance
\begin{equation*}
\langle v^{\vee} , C(w)B(x)A(y) u \rangle =
\langle\langle v^{\vee} , A(y)B(x)C(w) u \rangle\rangle |_{|w|>|x|>|y|},
\end{equation*}
for the last term of \eqref{eq:mar23_1}, 
where the right hand side is the expansion of \eqref{eq:mar23_2} into the region
$|w|>|x|>|y|$.
Since
the integration contours for this term ($C_{1,0}$ for $y$ and $C_{2,0}$ for $x$) are within 
this region, we can
replace $\langle v^{\vee} , C(w)B(x)A(y) u \rangle$
by \eqref{eq:mar23_2}.
By exactly the same reason we can replace the three-point functions
in all of the four terms of \eqref{eq:mar23_1}
by the same rational function
\eqref{eq:mar23_2}, and can interpret the factor $(x-w)^{-1}$ as a rational function
rather than its expansions.
Now all of the four terms of \eqref{eq:mar23_1} have the same integrand
which enables us to deform the contours as
\begin{align*}
&\langle v^{\vee} , A(w)_{(p)} (B(w)_{(-1)}C(w)) u \rangle\\
&=\frac{1}{(2 \pi \sqrt{-1})^2}
\oint_{C_{2,w}}\!\!\!d y \oint_{C_{1,w}}\!\!\!d x
\langle\langle v^{\vee} , A(y)B(x)C(w) u \rangle\rangle(y-w)^p (x-w)^{-1}\\
&=\frac{1}{(2 \pi \sqrt{-1})^2}
\oint_{C_{1,w}}\!\!\!d x \oint_{C_{x}}\!\!\!d y
\langle\langle v^{\vee} , A(y)B(x)C(w) u \rangle\rangle(y-w)^p (x-w)^{-1}\\
&+\frac{1}{(2 \pi \sqrt{-1})^2}
\oint_{C_{1,w}}\!\!\!d x \oint_{C_{w}}\!\!\!d y
\langle\langle v^{\vee} , A(y)B(x)C(w) u \rangle\rangle(y-w)^p (x-w)^{-1}.
\end{align*}
Now multiplying by $(z-w)^{-p-1}$ and summing over $p$ from $0$ to infinity,
we obtain
\begin{align}
&\langle v^{\vee} , \contraction{}{A}{(z)(}{BC}
A(z)(BC)(w)  u \rangle \nonumber\\
&=\frac{1}{(2 \pi \sqrt{-1})^2}
\oint_{C_{1,w}}\!\!\!d x \oint_{C_{x}}\!\!\!d y
\langle\langle v^{\vee} , A(y)B(x)C(w) u \rangle\rangle(z-y)^{-1} (x-w)^{-1}\nonumber\\
&+\frac{1}{(2 \pi \sqrt{-1})^2}
\oint_{C_{1,w}}\!\!\!d x \oint_{C_{w}}\!\!\!d y
\langle\langle v^{\vee} , A(y)B(x)C(w) u \rangle\rangle(z-y)^{-1} (x-w)^{-1}. \label{eq:mar23_3}
\end{align}
For the second term of this expression we consider the integral
\begin{equation}\label{eq:mar24_2}
I = \frac{1}{2 \pi \sqrt{-1}}
\oint_{C_{w}}\!\!\!d y
\langle\langle v^{\vee} , A(y)B(x)C(w) u \rangle\rangle(z-y)^{-1}.
\end{equation}
Note that the contour $C_{1,w}$ for $x$ can be deformed into 
those described by the difference of $C_{4,0}$ and $C_{1,0}$.
If $x \in C_{4,0}$ we have
\begin{align*}
I &= \frac{1}{2 \pi \sqrt{-1}}
\oint_{C_{3,0}}\!\!\!d y
\langle v^{\vee} , B(x)A(y)C(w) u \rangle (z-y)^{-1}\\
&- \frac{1}{2 \pi \sqrt{-1}}
\oint_{C_{2,0}}\!\!\!d y
\langle v^{\vee} , B(x)C(w)A(y) u \rangle (z-y)^{-1}\\
&= \frac{1}{2 \pi \sqrt{-1}}
\oint_{C_{w}}\!\!\!d y
\langle v^{\vee} , B(x)R(A(y)C(w)) u \rangle (z-y)^{-1}\\
&= \langle v^{\vee} , B(x) \contraction{}{A}{(z)}{C}A(z) C(w)u \rangle.
\end{align*}
Here we used Remark \ref{rem:mar29_8} for mutually local fields $A(y)$ and $C(w)$,
our definition of OPE \eqref{eq:july29_3}, and the
formula \eqref{eq:sep4_4} with $m=1$.
In the same way if  $x \in C_{1,0}$ we obtain
{$I = \langle v^{\vee} , \contraction{}{A}{(z)}{C}A(z) C(w) B(x)u \rangle$}.
Hence the second term of \eqref{eq:mar23_3} 
multiplied by $2 \pi \sqrt{-1}$ is written as
\begin{align}
\oint_{C_{4,0}}\frac{d x}{x-w}
\langle v^{\vee} , B(x) \contraction{}{A}{(z)}{C}A(z) C(w)u \rangle
&-
\oint_{C_{1,0}}\frac{d x}{x-w}
\langle v^{\vee} , \contraction{}{A}{(z)}{C}A(z) C(w) B(x)u \rangle\nonumber\\
&=
\oint_{C_{w}}\frac{d x}{x-w}
\langle v^{\vee} , R(B(x) \contraction{}{A}{(z)}{C}A(z) C(w))u \rangle.\label{eq:mar24_3}
\end{align}
Here we used Remark \ref{rem:mar29_8} for mutually local fields
$B(x)$ and $A(w)_{(p)}C(w)$
which is a consequence of Proposition \ref{prop:mar17_3}, 
and our definition of OPE \eqref{eq:jan27_2}.
This gives rise to the second term of the desired formula.
In a similar way, one can also derive the first term.
\qed

\subsection{Generalized Wick theorem for $\langle v^{\vee} , \contraction{}{(AB)}{(z)(}{C}(AB)(z)C(w)u \rangle$}\label{sec:4_4}
Now we present the main result of this paper in analytic formulation.
\begin{theorem}\label{th:4}
For any $u \in M, v^{\vee} \in M^{\vee}_u$ the following relation holds:
\begin{align*}
&\langle v^{\vee} , \contraction{}{(AB)}{(z)(}{C}
(AB)(z)C(w) u \rangle \\
&= \frac{1}{2 \pi \sqrt{-1}}
\oint_{C_w} \frac{d x}{z-x} \langle v^{\vee} , \left\{
\normOrdx{A(x) \contraction{}{B}{(x)}{C}B(x)C(w)} + R(B(x) \contraction{}{A}{(x)}{C}A(x) C(w)) \right\}u \rangle .
\end{align*}
Here $R$ represents the OPE given by \eqref{eq:jan27_2}.
\end{theorem}
%
\proof
We use the same notations in \eqref{eq:mar23_4} and \eqref{eq:mar23_2}.
By the definition of the residue product \eqref{eq:oct20_1},
the following relation is satisfied for any $p \in \N$:
\begin{align*}
&\langle v^{\vee} , (A(w)_{(-1)} B(w))_{(p)}C(w)) u \rangle\\
&={\rm Res}_{x=0} {\rm Res}_{y=0} 
\langle v^{\vee} , [[A(y), B(x)], C(w)] u \rangle (y-x)^{-1} (x-w)^{p}.
\end{align*}
Here the factor $ (y-x)^{-1}$ should be interpreted as its expansion
in the region $|y|>|x|$ (resp.~$|x|>|y|$) if the order of
the product of the operators is $A(y)B(x)$ (resp.~$B(x)A(y)$)
in each of the four terms in $[[A(y), B(x)], C(w)]$.
Then we have
\begin{align}
&{\rm Res}_{x=0} {\rm Res}_{y=0} 
\langle v^{\vee} , [[A(y), B(x)], C(w)] u \rangle (y-x)^{-1} (x-w)^{p}\nonumber\\
&=\frac{1}{(2 \pi \sqrt{-1})^2}
\oint_{C_{3,0}}\!\!\!d x \oint_{C_{4,0}}\!\!\!d y
\langle v^{\vee} , A(y)B(x)C(w) u \rangle (y-x)^{-1} (x-w)^{p}\nonumber\\
&-\frac{1}{(2 \pi \sqrt{-1})^2}
\oint_{C_{4,0}}\!\!\!d x \oint_{C_{3,0}}\!\!\!d y
\langle v^{\vee} , B(x)A(y)C(w) u \rangle (y-x)^{-1} (x-w)^{p}\nonumber\\
&-\frac{1}{(2 \pi \sqrt{-1})^2}
\oint_{C_{1,0}}\!\!\!d x \oint_{C_{2,0}}\!\!\!d y
\langle v^{\vee} , C(w)A(y)B(x) u \rangle (y-x)^{-1} (x-w)^{p}\nonumber\\
&+\frac{1}{(2 \pi \sqrt{-1})^2}
\oint_{C_{2,0}}\!\!\!d x \oint_{C_{1,0}}\!\!\!d y
\langle v^{\vee} , C(w)B(x)A(y) u \rangle (y-x)^{-1} (x-w)^{p}.\label{eq:mar23_5}
\end{align}
By exactly the same reason in the proof of Theorem \ref{th:leftwick},
we can replace the three-point functions
in all of the four terms of \eqref{eq:mar23_5}
by the same rational function
\eqref{eq:mar23_2}, and can interpret the factor $(y-x)^{-1}$ as a rational function
rather than its expansions.
Now all of the four terms of \eqref{eq:mar23_5} have the same integrand
which enables us to deform the contours as
\begin{align}
&\langle v^{\vee} , (A(w)_{(-1)} B(w))_{(p)}C(w)) u \rangle \nonumber\\
&=\frac{1}{(2 \pi \sqrt{-1})^2}
\oint_{C_{1,w}}\!\!\!d y \oint_{C_{w}}\!\!\!d x
\langle\langle v^{\vee} , A(y)B(x)C(w) u \rangle\rangle (y-x)^{-1} (x-w)^{p}\nonumber\\
&-\frac{1}{(2 \pi \sqrt{-1})^2}
\oint_{C_{1,w}}\!\!\!d x \oint_{C_{w}}\!\!\!d y
\langle\langle v^{\vee} , A(y)B(x)C(w) u \rangle\rangle (y-x)^{-1} (x-w)^{p}.
\label{eq:mar24_1a}
\end{align}
Here the first term comes from the first and third terms of 
\eqref{eq:mar23_5}, and the second term comes from its second and fourth terms.
Now multiplying by $(z-w)^{-p-1}$ and summing over $p$ from $0$ to infinity,
we obtain
\begin{align}
&\langle v^{\vee} , \contraction{}{(AB)}{(z)(}{C}
(AB)(z)C(w) u \rangle \nonumber\\
&=\frac{1}{(2 \pi \sqrt{-1})^2}
\oint_{C_{1,w}}\!\!\!d y \oint_{C_{w}}\!\!\!d x
\langle\langle v^{\vee} , A(y)B(x)C(w) u \rangle\rangle (y-x)^{-1} (z-x)^{-1}\nonumber\\
&-\frac{1}{(2 \pi \sqrt{-1})^2}
\oint_{C_{1,w}}\!\!\!d x \oint_{C_{w}}\!\!\!d y
\langle\langle v^{\vee} , A(y)B(x)C(w) u \rangle\rangle (y-x)^{-1} (z-x)^{-1}. \label{eq:mar24_1}
\end{align}
In view of the second term of \eqref{eq:mar24_1}
we consider the following integral
\begin{equation*}
J = -\frac{1}{2 \pi \sqrt{-1}}
\oint_{C_{w}}\!\!\!d y
\langle\langle v^{\vee} , A(y)B(x)C(w) u \rangle\rangle(y-x)^{-1}.
\end{equation*}
Then we see that after replacing $z$ by $x$
in \eqref{eq:mar24_2} the integral $I$ turns into this $J$.
Hence based on the same argument to deriving \eqref{eq:mar24_3}
one sees that the second term of \eqref{eq:mar24_1} is written as
\begin{equation*}
\frac{1}{2 \pi \sqrt{-1}}\oint_{C_{w}}\frac{d x}{z-x}
\langle v^{\vee} , R(B(x) \contraction{}{A}{(x)}{C}A(x) C(w))u \rangle,
\end{equation*}
that gives rise to the second term of the desired formula.

Now we consider the integral
\begin{equation*}
K = \frac{1}{2 \pi \sqrt{-1}}
\oint_{C_{w}}\!\!\!d x
\langle\langle v^{\vee} , A(y)B(x)C(w) u \rangle\rangle(y-x)^{-1} (z-x)^{-1},
\end{equation*}
in the first term of \eqref{eq:mar24_1}.
By using 
\begin{equation*}
(y-x)^{-1} (z-x)^{-1} = (z-y)^{-1}((y-x)^{-1} - (z-x)^{-1}),
\end{equation*}
and based on the same argument to deriving \eqref{eq:mar24_3}
we obtain

\begin{equation*}
K = (z-y)^{-1}
\left( \langle v^{\vee} , R(A(y) \contraction{}{B}{(y)}{C}B(y) C(w))u \rangle 
-\langle v^{\vee} , R(A(y) \contraction{}{B}{(z)}{C}B(z) C(w))u \rangle 
\right). \\
\end{equation*}
Hence
the first term of \eqref{eq:mar24_1} is written as
\begin{align*}
&
\frac{1}{2 \pi \sqrt{-1}}
\oint_{C_w} \frac{d y}{z-y} \left( \langle v^{\vee} , 
R(A(y) \contraction{}{B}{(y)}{C}B(y) C(w))
u \rangle
-\langle v^{\vee} , R(A(y) \contraction{}{B}{(z)}{C}B(z) C(w))u \rangle\right)
\\
&=\frac{1}{2 \pi \sqrt{-1}}
\oint_{C_w} \frac{d y}{z-y}\langle v^{\vee} , 
\bcontraction{}{A}{(y) (B(y}{)}
\contraction{A(y) (}{B}{(y)}{C}
A(y) (B(y) C(w))u \rangle\\
&+\frac{1}{2 \pi \sqrt{-1}}
\oint_{C_w} \frac{d y}{z-y} \langle v^{\vee} , 
\normOrd{\,A(y) \contraction{}{B}{(y)}{C} B(y)C(w)} 
u \rangle\\
&-\frac{1}{2 \pi \sqrt{-1}}
\oint_{C_w} \frac{d y}{z-y}\langle v^{\vee} , 
\bcontraction{}{A}{(y) (B(z}{)}
\contraction{A(y) (}{B}{(z)}{C}
A(y) (B(z) C(w))u \rangle\\
&-\frac{1}{2 \pi \sqrt{-1}}
\oint_{C_w} \frac{d y}{z-y} \langle v^{\vee} , 
\normOrd{\,A(y) \contraction{}{B}{(z)}{C} B(z)C(w)} 
u \rangle .
\end{align*}
The explicit expressions \eqref{eq:mar24_5} and \eqref{eq:jan27_1} tell us that
in the right hand side of this expression
the first and the third terms cancel out,
and the fourth term vanishes by itself.
Thus it remains only the second term 
that is equal to the first term of the desired formula by replacing $y$ by $x$.
\qed
%

\section{Discussion}\label{sec:6}
In the proof of Theorem \ref{th:4},
the contour deformation from \eqref{eq:mar23_5} to \eqref{eq:mar24_1a} itself
is fairly well-known.
For instance, Thielemans considered an
integral representation of a bit more general
nested residue product ((2.3.23) of Ref.~\citen{Th95}) 
which is essentially the same as
\begin{align*}
&\langle v^{\vee} , (A(w)_{(p)} B(w))_{(q)}C(w)) u \rangle\\
&={\rm Res}_{x=0} {\rm Res}_{y=0} 
\langle v^{\vee} , [[A(y), B(x)], C(w)] u \rangle (y-x)^{p} (x-w)^{q},
\end{align*}
with its right hand side rewritten in an integral expression
analogous to \eqref{eq:mar23_5}.
Then by using the above contour deformation he derived 
a specialization of the Borcherds identity ((2.3.24) of Ref.~\citen{Th95})
which is essentially the same as
\begin{align*}
(A(w)_{(p)}B(w))_{(q)}C(w)=
\sum_{i=0}^{\infty} (-1)^i {p \choose i} 
&\left( A(w)_{(-i+p)} (B(w)_{(q+i)}C(w))
\right.\nonumber\\
&\left.
-(-1)^p B(w)_{(q-i+p)} (A(w)_{(i)}C(w)) \right).
\end{align*}
Then by setting $p=-1$ and with \eqref{eq:mar16_1}
he derived its further specialization  ((3.3.18) of Ref.~\citen{Th95})
which is essentially the same as \eqref{eq:oct27x_1}.
We point out that, although he used the same contour deformation
to derive that,
no integral expression equivalent to \eqref{eq:main} has been
derived in Ref.~\citen{Th95}. 
In other words,
no {contour integral expression analogous to \eqref{eq:july3_1}
as a ``generating identity'' for \eqref{eq:oct27x_1} had been derived.
We admit that, having once obtained the integral expression \eqref{eq:main},
it is a rather straightforward task to derive 
\eqref{eq:mar31_2} or \eqref{eq:oct27x_1} from it.
However, the opposite 
direction of the derivation is fairly nontrivial. 

Sevrin et.~al.~had shown an identity ((A.7) of Ref.~\citen{STP88})
including the nested 
residue product $(A(w)_{(p)}B(w))_{(q)}C(w)$,
and claimed that for $p=-1$ it yields an integral expression 
of Wick's theorem ((A.8) of Ref.~\citen{STP88}).
Therefore, one might think that our formula \eqref{eq:main} 
has been already known.
However, their (A.8) is not equivalent to our \eqref{eq:main}
but to \eqref{eq:july3_1} for the opposite order contraction.
We suppose that their claim is wrong and
it is not their (A.7) but their (A.6) which yields
their integral expression (A.8) as the specialization,
since the latter is an identity including the nested 
residue product $A(w)_{(q)}(B(w)_{(p)}C(w))$.

While the contraction \eqref{eq:july3_2} can be regarded as
a sort of generating function of the residue products,
there is another type of generating function
which they call the {\em $\lambda$-bracket} \cite{Kac96}.
In this context, it is reasonable to give a discussion on the relation between our formula \eqref{eq:main} and the
{\em right noncommutative Wick formula} introduced in Refs.~\citen{BK03,KRW03}.
We will do this in Appendix \ref{app:A}.

From the viewpoint of practical calculations, 
our new formula \eqref{eq:main} is not indispensable for evaluating this contraction \cite{DMS97}.
However, we believe that our integral formula \eqref{eq:main} has
its own theoretical significance beyond practical usefulness because of its
simplicity.
As a formula in quantum field theory, 
we consider that
\eqref{eq:main} as well as \eqref{eq:july3_1} is written in an
expression which most physicists are more familiar with
than the other expressions.
Some examples for calculations of OPEs
by using \eqref{eq:main} will be illustrated in Appendix \ref{sec:5}.

In this paper we only dealt with the bosonic cases of integer conformal dimensions.
However, the generalization to the cases involving both bosonic and fermionic fields is straightforward \cite{T17}.

\vspace{5mm}
\noindent
{\it Acknowledgement}.
T.T. thanks Atsuo Kuniba, Yasuhiko Yamada,
Kotaro Watanabe, Michio Seto, and Kazuo Hosomichi for
valuable discussions.

\appendix
\section{Description by the $\lambda$-bracket}\label{app:A}
We show the relation between our formula \eqref{eq:main} and a
formula introduced in Refs.~\citen{BK03,KRW03}.
Firstly we introduce:
\begin{definition}[Ref.~\citen{Kac96}, (2.3.11)]
Given two series $A(w)$ and $B(w)$,
their \textbf{ $\lambda$-bracket} is
\begin{equation}
[A(w)_{\lambda} B(w)] =\sum_{m=0}^\infty
\lambda^{(m)} 
A(w)_{(m)} B(w),
\label{eq:mar20_1}
\end{equation}
where $\lambda^{(m)}  = \lambda^m/m!$, and $A(w)_{(m)} B(w)$ is the $m$-th residue product \red{defined as in}
\eqref{eq:oct20_1}.
\end{definition}
Now the left/right noncommutative Wick formulas
in Refs.~\citen{BK03,KRW03} are stated as follows.
\begin{proposition}
The following relations are satisfied.
\begin{align}
[A(w)_\lambda (BC)(w)] &=
\normOrdx{ \,[A(w)_\lambda B(w)] C(w)}+
\normOrdx{ B(w) [A(w)_\lambda C(w)]}\nonumber\\
&+
\int_0^\lambda [[A(w)_{\lambda} B(w)]_\mu C(w)] {\rm d} \mu,
\label{eq:mar29_1x}\\
[(AB)(w)_\lambda C(w) ]& =
\normOrdx{e^{\partial \frac{{\rm d}}{{\rm d} \lambda}} A(w) [B(w)_\lambda C(w)]}+
\normOrdx{e^{\partial \frac{{\rm d}}{{\rm d} \lambda}} B(w) [A(w)_\lambda C(w)]}\nonumber\\
&+
\int_0^\lambda [B(w)_\mu [A(w)_{\lambda - \mu} C(w)]] {\rm d} \mu.
\label{eq:mar20_2}
\end{align}
Here $\partial = \partial_w$ represents the differentiation that acts only on the 
nearest operator on its right side.
\end{proposition}
A proof of this proposition is given in Ref.~\citen{BK03}.
It is easy to see that the former identity \eqref{eq:mar29_1x}
is equivalent to \eqref{eq:mar31_1}, to which the integral formula \eqref{eq:july3_1}
 has been proved to be equivalent.
Here we present a heuristic proof of the latter identity \eqref{eq:mar20_2}
by showing that it is equivalent to the identity \eqref{eq:oct27x_1},
to which the integral formula
 \eqref{eq:main} has also been proved to be equivalent.

\proof
By definition, the left hand side of \eqref{eq:mar20_2} is equal to
\begin{equation*}
[(AB)(w)_\lambda C(w) ] =\sum_{q=0}^\infty
\lambda^{(q)} 
(A(w)_{(-1)}B(w))_{(q)}C(w),
\end{equation*}
that is a generating function of
the left hand side of the identity \eqref{eq:oct27x_1}.
On the other hand,
the first term of the right hand side of \eqref{eq:mar20_2} is rewritten as
\begin{align*}
\normOrdx{e^{\partial \frac{{\rm d}}{{\rm d} \lambda}} A(w) [B(w)_\lambda C(w)]}
&=\normOrdx{
\left(
\sum_{k=0}^\infty \frac{\left( \partial \frac{{\rm d}}{{\rm d} \lambda}\right)^k}{k!}
A(w)
\right)
\sum_{l=0}^\infty
\frac{\lambda^l}{l!}
B(w)_{(l)} C(w)
}\\
&=\sum_{k=0}^\infty \sum_{l=0}^\infty
\frac{1}{l!} \left( \frac{{\rm d}}{{\rm d} \lambda}\right)^k \lambda^l
\normOrdx{ \partial^{(k)} A(w) \left( B(w)_{(l)} C(w)
\right)
}\\
&=\sum_{l=0}^\infty \sum_{k=0}^l
\lambda^{(l-k)}
\normOrdx{ \partial^{(k)} A(w) \left( B(w)_{(l)} C(w)
\right)
}\\
&= \sum_{q=0}^\infty \lambda^{(q)}
\left(
\sum_{i=0}^\infty \normOrdx{ \partial^{(i)} A(w) \left( B(w)_{(q+i)} C(w)
\right)
}
\right).
\end{align*}
We see that the inner summation of the last expression is equal to 
the first term of the right hand side of the identity \eqref{eq:oct27x_1},
and that the same relation is satisfied between
their second terms
by exchanging $A$ and $B$.
Therefore,
the remaining task is to show that the same relation is satisfied between their third terms.
The third term of the right hand side of \eqref{eq:mar20_2} is rewritten as
\begin{align*}
I & :=\int_0^\lambda [B(w)_\mu [A(w)_{\lambda - \mu} C(w)]] {\rm d} \mu \\
&=\int_0^\lambda 
\sum_{k=0}^\infty \sum_{l=0}^\infty
\mu^{(l)} (\lambda - \mu)^{(k)} B(w)_{(l)} (A(w)_{(k)} C(w))
 {\rm d} \mu \\
&=\int_0^\lambda 
\sum_{q=0}^\infty \sum_{l=0}^q
\mu^{(l)} (\lambda - \mu)^{(q-l)} B(w)_{(l)} (A(w)_{(q-l)} C(w))
 {\rm d} \mu .
\end{align*}
By the definition of the residue product \eqref{eq:oct20_1},
the following relation is satisfied
\begin{align*}
B(w)_{(l)} (A(w)_{(q-l)}C(w)) 
&={\rm Res}_{x=0} {\rm Res}_{z=0} 
 [B(x), [A(z), C(w)]]   (z-w)^{q-l} (x-w)^{l},
\end{align*}
since $q \geq l \geq 0$.
By direct calculations we have
\begin{align*}
\sum_{l=0}^q
\mu^{(l)} (\lambda - \mu)^{(q-l)} (z-w)^{q-l} (x-w)^{l}
&= \frac{1}{q!} \left( \mu (x-w) + (\lambda - \mu) (z-w) 
\right)^q \\
&= \frac{1}{q!} \left( \lambda (z-w) + \mu (x-z) 
\right)^q,
\end{align*}
and 
\begin{align*}
\int_0^\lambda \left( \lambda (z-w) + \mu (x-z) 
\right)^q {\rm d} \mu
&=
\frac{\lambda^{q+1} \left( (x-w)^{q+1} - (z-w)^{q+1}
\right)}{(q+1) (x-z)} \\
&=\frac{\lambda^{q+1} }{q+1}
\sum_{i=0}^q (x-w)^i (z-w)^{q-i}.
\end{align*}
Again, by using the definition of the residue product, we have
\begin{align*}
{\rm Res}_{x=0} {\rm Res}_{z=0} 
 [B(x), [A(z), C(w)]]  (x-w)^{i} (z-w)^{q-i}
&= B(w)_{(i)} (A(w)_{(q-i)}C(w)) .
\end{align*}
Therefore
\begin{align*}
I &= \sum_{q=0}^\infty \sum_{i=0}^q \lambda^{(q+1)} B(w)_{(i)} (A(w)_{(q-i)}C(w))
= \sum_{q=0}^\infty \lambda^{(q)} \sum_{i=1}^q B(w)_{(i-1)} (A(w)_{(q-i)}C(w)).
\end{align*}
The inner summation of the last expression is equal to 
the third term of the right hand side of the identity \eqref{eq:oct27x_1}.
\qed

\section{Application to Operator Product Expansions}\label{sec:5}
\subsection{The OPE of $(TT)(z)$ and $T(w)$}\label{sec:5_1}
In order to illustrate the usefulness of our new formula \eqref{eq:main},
we shall present a few examples for calculations of the OPE.
For the sake of simplicity we omit the symbol $R$
for the radial ordering. 
Thus we write the generalized Wick theorems simply as
\eqref{eq:july3_1} and \eqref{eq:main},
and write the OPE simply as
\begin{equation}
A(y)B(z) = \contraction{}{A}{(y)}{B} A(y)B(z) + \normOrd{A(y)B(z)},
\end{equation}
instead of \eqref{eq:july29_3}.

First we consider the OPE of $(TT)(z)$ and $T(w)$,
where $T$ is the energy momentum tensor
satisfying the following OPE
\begin{equation}\label{eq:ttope}
T(x)T(w)  = \frac{{c/2}}{(x-w)^4} +\frac{2T(w)}{(x-w)^2} + \frac{\partial T(w)}{x-w}
+ \normOrd{T(x)T(w)}.
\end{equation}
From this formula and its derivative with respect to $w$
\begin{equation}
T(x)\partial T(w)  =  \frac{2c}{(x-w)^5} +\frac{4T(w)}{(x-w)^3} + \frac{3\partial  T(w)}{(x-w)^2} 
+ \frac{\partial^2 T(w)}{x-w} + \normOrd{T(x)\partial T(w)},
\end{equation}
we have
\begin{align}
T(x) \contraction{}{T}{(x)}{T}T(x) T(w) 
&= T(x) \left\{ \frac{{c/2}}{(x-w)^4} +\frac{2T(w)}{(x-w)^2} + \frac{\partial T(w)}{x-w} \right\} \nonumber\\
& =
\frac{3c}{(x-w)^6} +\frac{(c/2)T(x) + 8T(w)}{(x-w)^4} +\frac{5\partial  T(w)}{(x-w)^3} \nonumber\\
&+ \frac{2 \normOrd{T(x)T(w)}+ \partial^2 T(w)}{(x-w)^2} 
+ \frac{\normOrd{T(x)\partial T(w)}}{x-w}.
\end{align}
Then $\normOrd{T(x) \contraction{}{T}{(x)}{T}T(x) T(w)}$
is obtained from this expression by dropping all the terms
that do not contain the field $T(x)$.
Hence we have
\begin{align}
\normOrd{T(x) \contraction{}{T}{(x)}{T}T(x) T(w)}
&+T(x) \contraction{}{T}{(x)}{T}T(x) T(w) 
\nonumber\\
& =
\frac{3c}{(x-w)^6} +\frac{c T(x) + 8T(w)}{(x-w)^4} +\frac{5\partial  T(w)}{(x-w)^3} \nonumber\\
&+ \frac{4 \normOrd{T(x)T(w)}+ \partial^2 T(w)}{(x-w)^2} 
+ \frac{2 \normOrd{T(x)\partial T(w)}}{x-w}.\label{eq:mar24_6}
\end{align}
By using the Taylor expansion of $T(x)$ \red{at the point $x=w$} we obtain
\begin{align}
(\mbox{RHS of \eqref{eq:mar24_6}})
& =
\frac{3c}{(x-w)^6} +\frac{(8+c)T(w)}{(x-w)^4} +\frac{(5+c)\partial  T(w)}{(x-w)^3} \nonumber\\
&+ \frac{4 (TT)(w)+ (1+c/2)\partial^2 T(w)}{(x-w)^2} \nonumber\\
&+ \frac{(c/6) \partial^3 T(w) +  4 \normOrd{\partial T(w)T(w)}
 +  2 \normOrd{T(w)\partial T(w)}}{x-w} \nonumber\\
&+ (\mbox{regular terms}).
\end{align}
Applying \eqref{eq:main} to this expression amounts to
dropping the regular terms and replacing $x$ by $z$.
Hence we have
\begin{align}
\contraction{}{(TT)}{(z)(}{T}
(TT)(z)T(w) 
& =
\frac{3c}{(z-w)^6} +\frac{(8+c)T(w)}{(z-w)^4} +\frac{(5+c)\partial  T(w)}{(z-w)^3} \nonumber\\
&+ \frac{4 (TT)(w)+ (1+c/2)\partial^2 T(w)}{(z-w)^2} \nonumber\\
&+ \frac{(c/6) \partial^3 T(w) +  4 \normOrd{\partial T(w)T(w)}
 +  2 \normOrd{T(w)\partial T(w)}}{z-w}.
\end{align}
We note that the numerator of the last term can be written as
$(c-1) \partial^{(3)} T(w) + 3 \partial (TT)(w)$
by the following lemma.
Now one observes that our formula \eqref{eq:main} indeed
reproduces the known result (equation (6.214) of Ref.~\citen{DMS97}) for this OPE.
\begin{lemma}
\begin{equation*}
\normOrd{T(w)\partial T(w)} - \normOrd{\partial T(w)T(w)} = \partial^{(3)} T(w).
\end{equation*}
\end{lemma}
\proof
From the derivative of \eqref{eq:ttope} with respect to $x$
\begin{equation*}
\partial T(x)T(w)  = \frac{-2c}{(x-w)^5} - \frac{4T(w)}{(x-w)^3} - \frac{\partial T(w)}{(x-w)^2}
+ \normOrd{\partial T(x)T(w)},
\end{equation*}
one sees that $\partial T(w)_{(i-1)}T(w) = -2c I(w), -4T(w), -\partial T(w)$ for $i=5,3,2$
respectively and zero for the other $i \in \Z_{>0}$.
Then by the skew symmetry \eqref{eq:skewsym} we have
\begin{align*}
\normOrd{T(w)\partial T(w)} &= 
\sum_{i=0}^\infty (-1)^{i} \partial^{(i)} (\partial T(w)_{(i-1)}T(w)) \\
&= \normOrd{\partial T(w)T(w)} + \partial^{(3)} T(w).
\end{align*}
\qed
\subsection{Sugawara construction
for the quantized currents of affine Lie algebras}\label{sec:5_2}
Let $G$ be a simple Lie group and 
$\mathfrak{g}$ be its Lie algebra.
The quantized currents $J^a(z) \, (a=1,\dots , \mbox{dim} G)$ associated with $\mathfrak{g}$
are fields satisfying the following OPE (e.g. Refs.~\citen{DMS97,ES15})
\begin{equation}\label{eq:mar25_2}
J^a (z) J^b(w) = \frac{(k/2) \delta^{ab}}{(z-w)^2} + \frac{\sqrt{-1} f^{ab}_{\quad c} J^c(w)}{z-w}
+ \normOrd{J^a (z) J^b(w) }.
\end{equation}
Here $k$ is an integer called the {level}, and $ f^{ab}_{\quad c}$ denotes the structure constant
of $\mathfrak{g}$.
We adopted the convention of Ref.~\citen{ES15} where
the indices can be raised and lowered by
$g^{ab} = (1/2)\delta^{ab}$ and $g_{ab} = 2 \delta_{ab}$ respectively.
Then the structure constant $f^{abc}$ is anti-symmetric
with respect to any pair of indices.
They satisfy the relation
$\sum_{a,b} f^{abc} f_{abd} = 2 h_G^{\vee} \delta^c_d$
where $h_G^{\vee}$ is an integer called the {dual Coxeter number}.
In \eqref{eq:mar25_2} and hereafter, any repeated indices imply 
that summations are taken over them.

We consider the OPE of $(J^b J^b)(z)$ 
and $J^a(w)$.
From \eqref{eq:mar25_2} we have
\begin{align*}
&J^b(x) \contraction{}{J^b}{(x)}{J}{J^b(x)J^a(w)}=
J^b(x)\left\{
\frac{(k/2) \delta^{ba}}{(x-w)^2} + \frac{\sqrt{-1} f^{ba}_{\quad c} J^c(w)}{x-w}
\right\}
\\
&=\frac{(k/2) J^a(x)}{(x-w)^2}+\frac{\sqrt{-1} f^{ba}_{\quad c}}{x-w}
\left\{
\frac{(k/2) \delta^{bc}}{(x-w)^2} + \frac{\sqrt{-1} f^{bc}_{\quad d} J^d(w)}{x-w}
+ \normOrd{J^b (x) J^c(w) }
\right\}\\
&=\frac{(k/2) J^a(x)}{(x-w)^2}+
\frac{h_G^{\vee} J^a(w)}{(x-w)^2}+
\frac{\sqrt{-1} f^{ba}_{\quad c} \normOrd{J^b (x) J^c(w) }  }{x-w}.
\end{align*}
Here we used $f^{ba}_{\quad c}\delta^{bc}=0$ and 
$ f^{ba}_{\quad c}f^{bc}_{\quad d} = -h_G^{\vee} \delta^a_d$.
Then $\normOrd{J^b(x) \contraction{}{J^b}{(x)}{J}{J^b(x)J^a(w)} }$ is
obtained from this expression by dropping the second term.
Hence
\begin{align*}
\normOrd{J^b(x) \contraction{}{J^b}{(x)}{J}{J^b(x)J^a(w)} }
&+J^b(x) \contraction{}{J^b}{(x)}{J}{J^b(x)J^a(w)}
\\
&=\frac{k J^a(x)}{(x-w)^2}+
\frac{h_G^{\vee} J^a(w)}{(x-w)^2}+
\frac{2 \sqrt{-1} f^{ba}_{\quad c} \normOrd{J^b (x) J^c(w) }  }{x-w}\\
&= \frac{(k +  h_G^{\vee})J^a(w)}{(x-w)^2}
+\frac{k \partial J^a(w)+ 2 \sqrt{-1} f^{ba}_{\quad c} \normOrd{J^b (w) J^c(w) }  }{x-w}\\
&+ (\mbox{regular terms}).
\end{align*}
Therefore by \eqref{eq:main} we have
\begin{align}
\contraction{}{(J^bJ^b)}{(z)(}{J}
(J^bJ^b)(z)J^a(w) 
& =\frac{(k +  h_G^{\vee})J^a(w)}{(z-w)^2}
+\frac{k \partial J^a(w)+ 2 \sqrt{-1} f^{cb}_{\quad a} \normOrd{J^b (w) J^c(w) }  }{z-w}
\nonumber\\
&=(k +  h_G^{\vee})\left\{
\frac{J^a(w)}{(z-w)^2}+
\frac{\partial J^a(w) }{z-w}
\right\}.\label{eq:mar25_3}
\end{align}
Here we used $f^{ba}_{\quad c} = f^{cb}_{\quad a}$ and the following:
\begin{lemma}
$2 \sqrt{-1} f^{cb}_{\quad a} \normOrd{J^b (w) J^c(w) } = h_G^{\vee} \partial J^a(w)$.
\end{lemma}
\proof
From \eqref{eq:mar25_2} one sees that
$J^c(w)_{(i-1)} J^b(w) = (k/2) \delta^{cb}, \sqrt{-1} f^{cb}_{\quad d} J^d(w)$
for $i=2,1$ respectively and zero for the other $i \in \Z_{>0}$.
Then by the skew symmetry \eqref{eq:skewsym} we have
\begin{align*}
\normOrd{J^b (w) J^c(w)} &= 
\sum_{i=0}^\infty (-1)^{i} \partial^{(i)} (  J^c(w)_{(i-1)} J^b(w)    ) \\
&= \normOrd{J^c (w) J^b(w)} -\sqrt{-1} f^{cb}_{\quad d} \partial J^d(w).
\end{align*}
By multiplying $\sqrt{-1} f^{cb}_{\quad a}$ and summing over $b$ and $c$
we obtain the desired result.
\qed

%
Therefore we have observed that
our formula \eqref{eq:main} indeed reproduces
the well-known relation \eqref{eq:mar25_3} which implies that if we assume the 
following
\begin{equation}\label{eq:mar29_1}
T(z) = \frac{1}{k +h_G^{\vee} } (J^b J^b)(z),
\end{equation}
to be the energy momentum tensor, then
each current $J^a(w)$ behaves as a primary field of conformal dimension $1$.

Next we consider the calculation of the
OPE of $T(z)$ and $T(w)$ by using our formula
\eqref{eq:main}.
If we used the Wick theorem \eqref{eq:july3_1} instead of
\eqref{eq:main} 
in the above calculation
we would obtain \cite{BBSS88}
\begin{equation}\label{eq:mar29_3}
\contraction{}{(J^a)}{(z)(}{T}
J^a(z) T(w) = \frac{J^a(w)}{(z-w)^2}.
\end{equation}
On the other hand
by setting $a=b$ in \eqref{eq:mar25_2} and summing over $a$ from $1$ to dim $G$
we obtain
\begin{equation}
J^a (z) J^a(w) = \frac{(k/2) {\rm dim}G}{(z-w)^2}
+ \normOrd{J^a (z) J^a(w) }.
\end{equation}
Therefore
\begin{equation}
J^a(x) \contraction{}{(J^a)}{(x)(}{T}
J^a(x) T(w) =\frac{(k/2) {\rm dim}G}{(x-w)^4}
+ \frac{\normOrd{J^a (x) J^a(w) }}{(x-w)^2}.
\end{equation}
Then $\normOrd{J^a(x) \contraction{}{(J^a)}{(x)(}{T}J^a(x) T(w)}$ is
obtained by dropping its first term.
Hence
\begin{align*}
\normOrd{J^a(x) \contraction{}{(J^a)}{(x)(}{T}
J^a(x) T(w)} &+
J^a(x) \contraction{}{(J^a)}{(x)(}{T}
J^a(x) T(w)\\
&=\frac{(k/2) {\rm dim}G}{(x-w)^4}
+ \frac{2 \normOrd{J^a (x) J^a(w) }}{(x-w)^2}\\
&=\frac{(k/2) {\rm dim}G}{(x-w)^4}
+ \frac{2 \normOrd{J^a (w) J^a(w) }}{(x-w)^2}
+\frac{2 \normOrd{\partial J^a (w) J^a(w) }}{x-w}
\\
&+ (\mbox{regular terms}).
\end{align*}
Therefore by \eqref{eq:main} we have
\begin{align}
\contraction{}{(J^aJ^a)}{(z)(}{T}
(J^aJ^a)(z)T(w) 
& =\frac{(k/2) {\rm dim}G}{(z-w)^4}
+ \frac{2 \normOrd{J^a (w) J^a(w) }}{(z-w)^2}
+\frac{2 \normOrd{\partial J^a (w) J^a(w) }}{z-w}
\nonumber\\
&=(k +  h_G^{\vee})\left\{
\frac{c/2}{(z-w)^4}
+ \frac{2 T(w)}{(z-w)^2}
+\frac{\partial T(w) }{z-w}
\right\}.\label{eq:mar29_2}
\end{align}
Here we introduced the central charge $c$ 
defined as $c = k  {\rm dim}G/(k +  h_G^{\vee})$
and used the relation $\normOrd{\partial J^a (w) J^a(w)} = \normOrd{J^a (w) \partial J^a(w)}$
that is obtained by using the skew symmetry \eqref{eq:skewsym}.
Therefore we have observed that
our formula \eqref{eq:main} indeed reproduces
the well-known relation \eqref{eq:mar29_2} which implies that 
$T(z)$ in \eqref{eq:mar29_1}
satisfies the correct OPE for the energy momentum tensor \eqref{eq:ttope}.

%

\end{document}